\documentclass[A4,paper,12pt]{article}

\usepackage{amsmath}
\usepackage{amssymb}
\usepackage[cp1250]{inputenc}
\usepackage{placeins}
\usepackage{graphicx}

\newcommand{\im}{\mathrm{Im}\,}

\renewcommand{\d}{\mathsf{d}}
\newcommand{\opr}[1]{\mathsf{#1}}

\newcommand{\sobsp}{\mathcal{H}}

\newcommand{\supp}{\mathrm{supp}\,}
\newcommand{\laplace}{\triangle}

\renewcommand{\O}{\mathcal{O}}

\newenvironment{proof}{{\sl Proof:}\rm}{$\blacksquare$ \\ \smallskip}

\newtheorem{definition}{Definition}[section]
\newtheorem{lemma}[definition]{Lemma}
\newtheorem{corollary}[definition]{Corollary}
\newtheorem{theorem}[definition]{Theorem}

\newtheorem{proposition}[definition]{Proposition}

\begin{document}
\begin{flushleft}
\textbf{\large INTERLACED DENSE POINT AND ABSOLUTELY CONTINUOUS
SPECTRA FOR HAMILTONIANS WITH CONCENTRIC-SHELL SINGULAR
INTERACTIONS}\\
[2em]
{\large Pavel Exner and Martin Fraas} \\ [.2em]
{\small \emph{Nuclear Physics Institute, Czech Academy of
Sciences,
25068 \v{R}e\v{z} near Prague, \\
Doppler Institute, Czech Technical University, B\v{r}ehov\'{a} 7,
11519 Prague, Czechia \\ e-mail: exner@ujf.cas.cz,
fraas@ujf.cas.cz}
\\ [1em]

\begin{abstract}
We analyze the spectrum of the generalized Schr\"odinger operator
in $L^2(\mathbb{R}^\nu),\,\nu \geq 2$, with a general local,
rotationally invariant singular interaction supported by an
infinite family of concentric, equidistantly spaced spheres. It is
shown that the essential spectrum consists of interlaced segments
of the dense point and absolutely continuous character, and that
the relation of their lengths at high energies depends on the
choice of the interaction parameters; generically the p.p.
component is asymptotically dominant. We also show that for
$\nu=2$ there is an infinite family of eigenvalues below the
lowest band.
\end{abstract}

}

\textbf{keywords}: Schr\"odinger operators, singular interactions,
absolutely continuous spectrum, dense pure point spectrum

\end{flushleft}

\renewcommand{\theequation}{\arabic{section}.\arabic{equation}}

\setcounter{equation}{0}
%%%%%%%%%%%%%%%%%%%%%%%%%%%%%%%%%%%%%%%%%%%%%%%%%%
\section{Introduction}
\label{sec:introduction}

Quantum systems with the spectrum consisting of components of a
different nature attract attention from different points of view.
Probably the most important among them concerns random potentials
in higher dimensions --- a demonstration of existence of a
mobility edge is one of the hardest questions of the present
mathematical physics. At the same time, a study of specific
non-random systems can reveal various types of spectral behaviour
which differ from the generic type.

An interesting example among these refers to the situation where
the spectrum is composed of interlacing intervals of the dense
point and absolutely continuous character. A way to construct such
models using radially periodic potentials was proposed in
\cite{Hempel2}: since at large distances in such a system the
radial and angular variables ``almost decompose'' locally and the
radial part behaves thus essentially as one-dimensional there are
spectral intervals where the particle can propagate, with the gaps
between them filled densely by localized states.

To be specific consider, e.g., the operator $\opr{t} = -\d^2 / \d
x^2 + q(x)$ on $L^2(\mathbb{R})$ with $q$ bounded and periodic. By
the standard Floquet analysis the spectrum of $\opr{t}$ is purely
absolutely continuous consisting of a family of bands,
$\sigma(\opr{t}) = \bigcup_{k=0}^N[E_{2k},E_{2k + 1}]$,
corresponding to a strictly increasing, generically infinite
sequence $\{E_k\}_{k=0}^N$. Suppose now that the potential is
mirror-symmetric, $q(x) = q(-x)$, and consider the operator
  %----
  $$ \opr{T} = -\laplace + q(|x|) $$
  %-----
on $L^2(\mathbb{R}^\nu),\,\nu \geq 2$. It was shown in
\cite{Hempel2} that the essential spectrum of $\opr{T}$ covers the
half-line $[E_0,\,\infty)$, being absolutely continuous in the
spectral bands of $\opr{t}$ and dense pure point in the gaps
$(E_{2k-1},E_{2k}),\:k=1,\dots N$.

The well-known properties of one-dimensional Schr\"odinger
operators tell us that the dense point segments in this example
shrink with increasing energy at a rate determined by the
regularity of the potential. If we replace the bounded $q$ by a
family of $\delta$ interactions, the segment lengths tend instead
to a positive constant \cite{Exner2} , nevertheless, the
absolutely continuous component still dominates the spectrum at
high energies.

The aim of this paper is to investigate a similar model in which a
family of concentric, equally spaced spheres supports generalized
point interactions with identical parameters. We will demonstrate
that the interlaced spectral character persists and, depending on
the choice of the parameters, each of the components may dominate in
the high-energy limit, or neither of them. Specifically, the ratio
of the adjacent \emph{pp} and \emph{ac} spectral segments,
$(E_{2k}-E_{2k-1})/(E_{2k+1}-E_{2k})$, has three possible types of
behaviour, namely like $\mathcal{O}(k^{\mu})$ with $\mu=0, \pm 1$.
What is more, in the generic case we have $\mu=1$ so the dense point
part dominates, which is a picture very different from the
mobility-edge situation mentioned in the opening. Apart of this main
result, we are going to show that the interesting result about
existence of the so-called ``Welsh eigenvalues'' in the
two-dimensional case\cite{Welsh, Schmidt2} also extends to the case
of generalized point interactions.

\setcounter{equation}{0}
%%%%%%%%%%%%%%%%%%%%%%%%%%%%%%%%%%%%%%%%%%%%%%%%%%
\section{The model}
\label{sec:model}

As we have said, we are going to investigate generalized
Schr\"odinger operators in $\mathbb{R}^\nu,\:\nu\geq2$, with
spherically symmetric singular interaction on concentric spheres,
the radii of which are supposed to be $R_n = nd + d/2,\: n \in
\mathbb{N}$. It is important that the system is radially periodic,
hence the interactions on all the spheres are assumed to be the
same. In view of the spherical symmetry we may employ the
partial-wave decomposition: the isometry $\mathsf{U}\,:\,
L^2((0,\,\infty),r^{\nu-1}dr)\, \rightarrow \, L^2(0,\,\infty)$
defined by  $\mathsf{U}f(r) = r^{\frac{\nu-1}{2}}f(r)$ allows us
to write $L^2(R^\nu) = \bigoplus_{l\in\mathbb{N}_0}
\mathsf{U}^{-1}L^2(0,\,\infty)\otimes S_l$, where $S_l$ is the
$l$-th eigenspace of the Laplace-Bertrami operator on the unit
sphere. The operator we are interested in can be then written as
    %-------
   \begin{equation}
   \opr{{H}}_\Lambda := \bigoplus_l \mathsf{U}^{-1}
   \opr{{H}}_{\Lambda,l}\mathsf{U} \otimes \opr{I}_l ,
   \label{eq:opr}
   \end{equation}
   %--------
where $ \opr{I}_l $ is the identity operator on $ S_l$ and the
$l$-th partial wave operators
    %--------
\begin{equation}
    \opr{{H}}_{\Lambda,l} := -\frac{d^2}{d\,r^2} +
    \frac{1}{r^2}\left[\frac{(\nu-1)(\nu-3)}{4} +
   l(l+\nu-2)\right]
   \label{eq:ParOpr}
\end{equation}
    %--------
are determined by the boundary conditions\footnote{For relations
of these conditions to the other standard parametrization of the
generalized point interaction, $(U-I)F(R_n) + i(U+I)F'(R_n)=0$,
see \cite{exner1} } at the singular points $R_n$,
    %----------------
\begin{equation}
 \left(\begin{array}{c}
             f(R_n +) \\
             f'(R_n +)
         \end{array}\right) =
   e^{i \chi} \left( \begin{array}{cc}
\gamma & \beta\\
\alpha & \delta
 \end{array} \right)\left(\begin{array}{c}
             f(R_n -) \\
             f'(R_n -)
         \end{array}\right);
         \label{eq:BoundCon}
\end{equation}
%-------------
in the transfer matrix $\Lambda:=e^{i \chi} \binom{\gamma \,
\beta}{\alpha \, \delta} $ the parameters $\alpha, \beta, \gamma,
\delta$ are real and satisfy the condition $\alpha \beta - \gamma
\delta = -1$. In other words, the domain of the selfadjoint
operator $\opr{H}_{\Lambda,l}$ is
%---------
\begin{eqnarray}
    &D(\opr{H}_{\Lambda,l}) = \Big\{f \in L^2(0,\,\infty):\: \,f,\,f' \in
    AC_\mathrm{loc} \big((0,\,\infty) \setminus \cup_n\{R_n\} \big);& \nonumber \\
    &  -f'' + \frac{1}{r^2}\left[\frac{(\nu-1)(\nu-3)}{4} +
   l(l+\nu-2)\right]f \in L^2(0,\,\infty);& \nonumber \\
  & F(R_n+) = \Lambda F(R_n-) \Big\},&
     \label{eq:5.rest}
   \end{eqnarray}
%---------
where the last equation is a shorthand for the boundary conditions
(\ref{eq:BoundCon}). If the dimension $ \nu \leq 3 $ we have to
add a condition for behaviour of $f \in D(\opr{H}_l)$ at the
origin: for $ \nu=2,\,l=0$ we assume that $ \lim_{r \to 0+}
[\sqrt{r} \ln{r}]^{-1} f(r) =0,$ and for $\nu=3,\,l=0$ we replace
it by $ f(0+) = 0$. Since the generalized point interaction is
kept fixed, we will mostly drop the symbol $\Lambda$ in the
following.

\setcounter{equation}{0}
%%%%%%%%%%%%%%%%%%%%%%%%%%%%%%%%%%%%%%%%%%%%%%%%%%
\section{Generalized Kronig-Penney model}
\label{sec:KP}

As in the regular case the structure of the spectrum is determined
by the underlying one-dimensional Kronig-Penney model. We need its
generalized form where the Hamiltonian acts as the one dimensional
Laplacian except at the interaction sites, $x_n := nd + d/2,\, n
\in \mathbb{Z}$, where the wave functions satisfy boundary
conditions analogous to (\ref{eq:BoundCon}). To be explicit we
consider the four-parameter family of self-adjoint operators
%-------
\begin{equation}
\label{eq:KP1D}
  \opr{h}_\Lambda f:= -f'', \;\,
  D(\opr{h}_\Lambda) = \Big\{ f \in \sobsp^{2,\,2} \big(\mathbb{R}
  \setminus
  \cup_n \{x_n\}\big):\: F(x_n+) = \Lambda F(x_n-) \Big\}
\end{equation}
%------
with the matrix $\Lambda$ same as above (without loss of
generality we may assume $\chi = 0$ because it is easy to see that
operators differing by the value of $\chi$ are isospectral).
Spectral properties of this model were investigated in
\cite{Cheon, Exner3} where it was shown that the following three
possibilities occur:
%------
\begin{enumerate}
\renewcommand{\labelenumi}{(\roman{enumi})}
\item \emph{the $\delta$-type:} $\beta=0$ and $\gamma = \delta
=1$. In this case the gap width is asymptotically constant; it
behaves like $2|\alpha|d^{-1} + \mathcal{O}(n^{-1})$ as the band
index $n\to\infty$. This is the standard Kronig-Penney model.
%------
\item \emph{the intermediate type:} $\beta = 0$ and $|\gamma +
\delta|>2$. Now the quotient of the band width to the adjacent gap
width is asymptotically constant behaving as $\arcsin(2|\delta +
\gamma|^{-1})/\arccos(2|\delta + \gamma|^{-1}) + \O(n^{-1})$.
%------
\item \emph{the $\delta'$-type:} the generic case, $\beta \ne 0$.
In this case the band width is asymptotically constant; it behaves
like $8|\beta d|^{-1} + \O(n^{-1})$ as $n\to\infty$.
\end{enumerate}
%------
Recall that these types of spectral behaviour correspond to
high-energy properties of a single generalized point interaction
as manifested through scattering, resonances \cite{Exner1}, etc.

There is one more difference from standard Floquet theory which we
want to emphasize. It is well known \cite{Weid2} that in the
regular case the spectral edge $E_0$ corresponds to a symmetric
eigenfunction. In the singular case this is no longer true; one
can check easily the following claim.
%------
\begin{proposition}
\label{def:Sym} Let $u$ be an $\alpha$-periodic solution of the
equation $-u'' = E_0 u$ on $(-d/2,\,d/2)$ with $U(x_n+) = \Lambda
U(x_n-)$, where $E_0:=\inf \sigma(\opr{h}_\Lambda)$. Then $u$ is
periodic for $\beta \geq 0$ and antiperiodic for $\beta < 0$.
\end{proposition}

To finish the discussion of the one-dimensional comparison
operator, let us state three auxiliary results which will be
needed in the next section.
%------
\begin{lemma}
There is a constant $C>0$ such that for every function $u$ in the
domain of the operator $\opr{h}_{\Lambda} $ it holds that
 % ------------- %
 \begin{equation}
\|u'\| \leq C(\|\opr{h}_{\Lambda}u\| + \|u\|) \,. \label{eq:1}
\end{equation}
 % ------------- %
 \label{def:OdhadDer}
\end{lemma}
%------
\begin{proof}
We employ Redheffer's inequality \cite{Redheffer} which states
that
%---------
$$ \int\limits_a^b |u'(x)|^2 \,\d x \leq C'\left(
\int\limits_a^b |u''(x)|^2 \,\d x + \int\limits_a^b |u(x)|^2 \,\d x\right) $$
%---------
holds for any $u$ twice differentiable in an interval $[a,\,b]$
and some $C'>0$; then we get an inequality similar to (\ref{eq:1})
for the squares of the norms by summing up these inequalities with
$a=x_n,\,b=x_{n+1}$, and the sought result with $C=2C'$ follows
easily.
\end{proof}

%---------
\begin{lemma}
The set of functions from $D(\opr{h}_\Lambda)$ with a compact
support is a core of the operator $\opr{h}_\Lambda$.
\label{def:core}
\end{lemma}
%---------
 \begin{proof}
To a given $u \in D(\opr{h}_\Lambda)$ and $\varepsilon >0$ we will
construct an approximation function $u_\varepsilon \in
D(\opr{h}_\Lambda)$ which is compactly supported to the right,
i.e. it satisfies $\sup\, \mathrm{supp\,} u_\varepsilon<\infty$,
and
 %------
 $$
\int_{\mathbb{R}} \big(|u - u_\varepsilon|^2 + |u'' -
u_\varepsilon''|^2\big)(t)\, \d t \leq \varepsilon\,.
$$
%-------
Given $x\in\mathbb{R}$ and $d>0$ we can employ for a function $v
\in \sobsp^{2,2}(x,x+d)$ the Sobolev embedding,
%-----
$$ |v(x)|^2 + |v'(x)|^2 \leq
\sup\limits_{t \in [x,\,x+d]}|v(t)|^2 + |v'(t)|^2 \leq
C_1\int_x^{x+d}(|v|^2 + |v''|^2)(t)\, \d t$$
%-----
with a constant $C_1$ which depends on $d$ but not on $x$. Let us
take next a pair of functions, $\phi_i \in C^\infty(0,d),\,i=1,2$,
such that they satisfy $\phi_1(0) = \phi_2'(0) =1$ and
$\phi_1'(0)=\phi_2(0)=\phi_i(d)=\phi_i'(d)=0$. Denote by $M_i$ the
maximum of $|\phi_i(t)|^2 + |\phi_i''(t)|^2$ and put
$M:=\max\{M_1,\,M_2\}$; then it holds
%----
$$
\int_0^{d}(|a\phi_1 + b\phi_2|^2 + |a\phi_1''+b\phi_2''|^2)(t)\,
\d t \leq 2Md(a^2 + b^2).
$$
%----
In view of the  assumption made about the function $u$ we can find
$n$ such that $\int_{x_n}^\infty(|u|^2+|u''|^2)(t) \d t \leq
\tilde\varepsilon := \varepsilon/(2 + 8MdC_1)$ and define
%------
$$
u_\varepsilon(x) := \left\{ \begin{array}{lcr}
                        u(x) & \;\;\mathrm{if}\;\; & x \le x_n \\
                        u(x_n+)\phi_1(x) + u'(x_n+)\phi_2(x) &
                        \;\;\mathrm{if}\;\;
                        & x \in (x_n,\,x_n+d) \\
                        0 & \;\;\mathrm{if}\;\; & x \ge x_n +d
                        \end{array} \right.
$$
Then $u_\varepsilon$ belongs to $ D(\opr{h}_\Lambda) $ being
compactly supported to the right and
%-----
\begin{multline}
\int_{\mathbb{R}} (|u-u_\varepsilon|^2 + |u'' -
u_\varepsilon''|^2)(t) \,\d t \leq 2\int_{x_n}^\infty (|u|^2
+|u''|^2+|u_\varepsilon|^2 + |u_\varepsilon''|^2)(t) \,\d t \\
\leq 2\int_{x_n}^\infty (|u|^2 +|u''|^2)(t)\, \d t + 8Md
(|u(x_n)|^2 + |u'(x_n)|^2) \\ \leq (2 + 8Md C_1)\int_{x_n}^\infty
(|u|^2 +|u''|^2)(t) \, \d t \leq (2 + 8MdC_1) \tilde\varepsilon =
\varepsilon\,. \nonumber
\end{multline}
%----
Furthermore, one can take this function $u_\varepsilon$ and
perform on it the analogous construction to get the support
compact on the left, arriving in this way at a compactly supported
$\tilde u_\varepsilon$ such that
 %------
 $$
\int_{\mathbb{R}} \big(|u - \tilde u_\varepsilon|^2 + |u'' -
\tilde u_\varepsilon''|^2\big)(t)\, \d t \leq 2\varepsilon\,,
$$
%-------
and since $\varepsilon$ was arbitrary by assumption the lemma is
proved.
  \end{proof}

The last one is a simple observation, which is however the main
tool for conversion of the proofs in the regular case to their
singular counterparts.

\begin{lemma}
  Let $u,\,v \in D(\opr{h}_\Lambda)$, then the Wronskian
  \begin{equation}
    W[\bar{u},\,v](x) := \bar{u}(x) v'(x) - \bar{u}'(x) v(x)
    \label{eq:Wrons}
  \end{equation}
  is a continuous function of $x$ on the whole real axis.
  \label{def:Wrons}
\end{lemma}
%-------
\begin{proof}
The condition $\alpha \beta - \gamma \delta = -1$ for the transfer
matrix $\Lambda$ is equivalent to $\Lambda^* \sigma_2 \Lambda =
\sigma_2$, where $\sigma_2$ is the second Pauli
matrix\cite{Cheon}. Then we have
%---------
\begin{eqnarray*}
\lefteqn{W[\bar{u},\,v](x_n+) = i U^{*}(x_n+) \sigma_2 V(x_n+)  =
i(\Lambda U(x_n-))^* \sigma_2 \Lambda V(x_n-)} \\ && \phantom{AAA}
= iU^*(x_n-) \sigma_2 V(x_n-) = W[\bar{u},\,v](x_n-)\,,
\phantom{AAAAAAAAAA}
\end{eqnarray*}
%---------
which concludes the proof.
\end{proof}
%---------
The way in which we are going to employ this result is the
following. Suppose we have \emph{real-valued} functions
$u_0,\,v_0,\,u$ which are $\sobsp^{2,\,2}$ away from the points
$x_n$ and satisfy the boundary conditions (\ref{eq:BoundCon}) at
them. Let, in addition, $W[u_0,\,v_0]$ be nonzero -- in the
applications below this will be true as $u_0,\,v_0$ will be
linearly independent generalized eigenfunctions of
$\opr{h}_\Lambda$ -- then by the lemma the vector function
%-------
\begin{equation}
  y = \left[ \begin{array}{cc}
                u_0 & v_0 \\
                u_0'& v_0'
             \end{array}\right]^{-1}
             \left(\begin{array}{c}
                u \\
                u'
             \end{array}\right) = W[u_0,\,v_0]^{-1}
             \left(\begin{array}{c}
                v_0'u - v_0 u' \\
                -u_0'u + u_0 u'
             \end{array}\right)
             \label{eq:Wcon}
\end{equation}
%-------
is continuous everywhere including the points $x_n$.

\setcounter{equation}{0}
%%%%%%%%%%%%%%%%%%%%%%%%%%%%%%%%%%%%%%%%%%%%%%%%%%
\section{The essential spectrum}
\label{sec:ess_spectrum}

Now we are going to demonstrate the spectral properties of
$\opr{H}_\Lambda$ announced in the introduction. We follow the
ideology used in the regular case~\cite{Hempel1, Hempel2},
localizing first the essential spectrum and finding afterwards the
subsets where it is absolutely continuous. In view of the partial
wave decomposition (\ref{eq:opr}) it is natural to start with the
partial wave operators $\opr{H}_l$.

The essential spectrum is stable under a rank one perturbation,
hence adding the Dirichlet boundary condition at a point $a>0$ to
each of the operators $\opr{H}_l,\,\opr{h}_\Lambda$ we do not
change their essential spectrum. Moreover, multiplication by $C
x^{-2}$ is a relatively compact operator on $L^2(a,\infty)$, thus
the essential spectra of the said operators coincide,
%---------
\begin{equation}
\sigma_{ess}(\opr{H}_l) = \sigma_{ess}(\opr{h}_\Lambda)
\label{eq:EssSpectra}.
\end{equation}
%---------
With this prerequisite we can pass to our first main result.
%---------
\begin{theorem}{The essential spectrum of the operator
(\ref{eq:opr}) is equal to}
%-------
\begin{equation}
\sigma_{ess}(\opr{H}_\Lambda) = [\inf
\sigma_{ess}(\opr{h}_\Lambda),\infty)\,.\label{eq:3.2.1}
\end{equation}
\label{def:ess}
%-------
\end{theorem}
%-------
The idea of the proof is the same as in \cite{Hempel1}: first we
check that $\inf \sigma_{ess}(\opr{H}_\Lambda)$ cannot be smaller
then $\inf \sigma_{ess}(\opr{h}_\Lambda)$, after that we show that
$\sigma_{ess}(\opr{H}_\Lambda)$ contains the whole interval $[\inf
\sigma_{ess}(\opr{h}_\Lambda),\infty)$.
%-------
\begin{proposition}{Under the assumptions stated we have}
%-------
\begin{equation}
\inf \sigma_{ess}(\opr{H}_\Lambda) \geq \inf
\sigma_{ess}(\opr{h}_\Lambda)\,. \label{def:p1}
\end{equation}
%-------
\end{proposition}
%-------
\begin{proof}
If $\nu > 2$ we infer from equations (\ref{eq:EssSpectra}),
(\ref{eq:EssSpectrab}) that
%--------
$$
\inf \sigma_{ess}(\opr{H}_\Lambda) \geq \inf_l \inf
\sigma(\opr{H}_l) = \inf \sigma_{ess}(\opr{h}_\Lambda)\,;
$$
%-------
notice that with the exception  of the case $\nu=2,\,l=0$ the
centrifugal term in the partial waves operators (\ref{eq:ParOpr})
is strictly positive, and consequently, the mini-max principle
implies
%---------
\begin{equation}
\inf \sigma(\opr{H}_l) \geq \inf \sigma(\opr{h}_\Lambda) = \inf
\sigma_{ess}(\opr{H}_l) \geq \inf \sigma(\opr{H}_l).
\label{eq:EssSpectrab}
\end{equation}
%---------
For $\nu =2$ the argument works again, we have just to be a little
more cautious and consider in the first partial wave the infimum
over the essential spectrum only.
\end{proof}

\begin{proposition}
%------
$\;\sigma_{ess}(\opr{H}_\Lambda) \supset [\inf
\sigma_{ess}(\opr{h}_\Lambda),\infty]$.
%------
\end{proposition}
%------
\begin{proof}
The idea is to employ Weyl criterion \cite{Weid}. Let $\lambda_0
\in \sigma_{ess}(h_\Lambda)$ and $ \lambda > 0 $, then we have to
show that for every $\varepsilon >0$ there exists a function
%-------
$$
 \phi \in D(\opr{H}_{\Lambda})\quad \mbox{satisfying}\quad
\|(\opr{H}_{\Lambda} - \lambda_0 - \lambda)\phi\| \leq \varepsilon
\|\phi\| .
$$
%-------
Basic properties of the essential spectrum together with Lemma
\ref{def:core} provide us with a compactly supported $u \in
D(\opr{h}_\Lambda)$ such that $\|u''-\lambda_0u\| \leq \frac12
\varepsilon$. in view of the periodicity we may suppose that $\supp
u \subset (0,\,L)$. Next we are going to estimate $\lambda$ by the
repulsive centrifugal potential in a suitably chosen partial wave.
Putting $l_R := [\sqrt{\lambda R}\,]$ we have
%-------
$$
\frac{1}{r^2}\left[\frac{(\nu-1)(\nu-3)}{4} +
   l_R(l_R+\nu-2)\right] = \lambda +
   \O(R^{-1})\quad \mathrm{for} \quad r\in[R,\,R+L]
$$
%--------
as $R \to \infty$, hence choosing $R$ large enough one can achieve
that
%--------
$$
\sup_{r\in[R,\,R+L]}
\left|\frac{1}{r^2}\left[\frac{(\nu-1)(\nu-3)}{4} +
   l_R(l_R+\nu-2)\right] - \lambda\right| \leq
   \frac12 {\varepsilon}\,.
   $$
%--------
Next we employ the partial wave decomposition, considering a unit
vector $Y \in S_{l_{n_\varepsilon}}$ and putting
$\phi(x):=\opr{U}^{-1} u(|x|-R)Y(x/|x|) $. It holds obviously
$\phi \in D(\opr{H}_\Lambda),\:\|\phi\| =\|u(\cdot-R)\| $,
and\footnote{For simplicity we allow ourselves the licence to
write $\|f\|\equiv \|f(\cdot)\|=\|f(r)\|$ in the following
formula.}
%----
\begin{eqnarray*}
\lefteqn{\|\opr{H}_\Lambda \phi - (\lambda_0 + \lambda) \phi\| =
\|\opr{H}_{l_R} u(r-R) - (\lambda_0 + \lambda)u(r-R)\|} \\ && \leq
\|u''(r-R) - \lambda_0 u(r-R)\|  \\ && +
\left\|\left(\frac{1}{r^2}\left[\frac{(\nu-1)(\nu-3)}{4} +
l_R(l_R+\nu-2)\right] - \lambda \right)u(r-R)\right\| \leq
\varepsilon \|\phi\|,
\end{eqnarray*}
%-----
which concludes the proof.
\end{proof}

Once the essential spectrum is localized, we can turn to its
continuous component. In view of the decomposition (\ref{eq:opr})
we have to describe the continuous spectrum in each partial wave
and the results for $\opr{H}_\Lambda$ will immediately follow;
recall that the essential spectrum of $\opr{H}_l$ consists of the
bands of the underlying one-dimensional operator
$\opr{h}_\Lambda$. Our strategy is to prove that the transfer
matrix --- defined in the appendix, Sec.~\ref{s:append} below
--- is bounded inside the bands, which implies that the spectrum
remains absolutely continuous \cite{GilPea,Simon}. The following
claim is a simple adaptation of the Lemma~2 from \cite{Hempel2} to
the singular case.

\begin{lemma}
Let $(a,\,b)$ be the interior of a band of the operator
$\opr{h}_\Lambda$ in $L_2(\mathbb{R})$. Let further $K \subset
(a,\,b)$ be a compact subinterval, $c \in \mathbb{R}$, and $x_0 >
0$. Then there is a number $C>0$ such that for every $\lambda \in
K$ any solution $u$ of
%------
\begin{equation}
-u'' + \frac{c}{r^2}u = \lambda u, \quad u \in D(\opr{h}_\Lambda)
\label{eq:pur.1}
\end{equation}
%------
with the normalization
%------
\begin{equation}
|u(x_0)|^2 + |u'(x_0)|^2 = 1 \label{eq:pur.2}
\end{equation}
%------
satisfies in $(x_0,\infty)$ the inequality
%------
\begin{equation}
 |u(x)|^2 + |u'(x)|^2 \leq C \,. \label{eq:pur.3}
\end{equation}
%--------
\label{def:pur.1}
\end{lemma}
%--------
\begin{proof}
For a fixed $\lambda \in K$ the equation $\opr{h}_\Lambda w =
\lambda w$ has two real-valued, linearly independent solutions, $
u_0 = u_0(\cdot,\,\lambda)$ and $v_0 = v_0(\cdot,\,\lambda)$, such
that $u_0,\,v_0 \in D(\opr{h}_\Lambda)$ and the functions
$|u_0|,\,|u'_0|,\,|v_0|,\,|v'_0| $ are periodic, bounded, and
continuous with respect to $\lambda$, cf.~\cite{Weid2}. Without
loss of generality we may assume that the determinant of the
matrix
%------
  $$
   Y = \left[\begin{matrix}
                u_0 & v_0 \\
                u'_0 & v'_0
        \end{matrix}\right]
   $$
%------
equals one; note that $u_0,\,v_0$ are real-valued and hence $\det
Y$ is continuous at the singular points in view of to the Lemma
\ref{def:Wrons}. It is also nonzero, hence to any solution $u$ of
(\ref{eq:pur.1}) we can define the function
%------
$$
y := Y^{-1} \left[ \begin{matrix}
                        u \\
                        u'
                    \end{matrix} \right]
$$
%-------
which satisfies
%-------
\begin{equation}
y'= A y \quad \text{on every interval} \quad (na,\,(n+1)a),
\label{eq:pur.8}
\end{equation}
%-------
where the  the matrix $A$ is given by
%-------
$$
A := -\frac{c}{x^2} \left[ \begin{matrix}
                                u_0 v_0 & v_0^2 \\
                                -u_0^2 & -u_0v_0
                            \end{matrix} \right],
$$
%-------
being integrable away of zero. By a straightforward calculation we
get
%-------
$$
y = \left[\begin{matrix}
                v'_0 u - v_0 u' \\
                -u'_0 u + u_0 u'
            \end{matrix} \right]
$$
%-------
and using Lemma~\ref{def:Wrons} again we infer that $y$ is
continuous at the singular points. Consequently,
%-----
$$
y(x) = \exp\bigg\{\int\limits_{x_0}^x A(t)\, \d t\bigg\} y(x_0)
$$
%-----
is a solution of (\ref{eq:pur.8}) and following \cite{Hempel2} we
arrive at the estimates
%-----
$$
\frac{1}{2}(|y|^2)'\leq |(y,\,y')| \leq \|A\||y|^2
$$
%-----
 and so
%-----
$$
|y(x)|^2 \leq |y(x_0)|^2 \exp\bigg\{2 \int\limits_{x_0}^x \|A\|
(t)\,\d t \bigg\} \leq |Y^{-1}(x_0)|^2 \exp\bigg\{ 2
\int\limits_{x_0}^\infty \|A\| (t)\,\d t \bigg\}
$$
%-----
for $x \geq x_0$ and every solution of (\ref{eq:pur.1}) with the
normalization (\ref{eq:pur.2}). From
%------
$$
\left[ \begin{matrix}
                u(x) \\
                u'(x)
            \end{matrix} \right] = Y(x) Y^{-1}(x_0) \left[ \begin{matrix}
                u(x_0) \\
                u'(x_0)
            \end{matrix} \right] + \int\limits_{x_0}^{x} Y(x) A(t) y(t)
            \d t
$$
%------
we then infer that the function $|u(\cdot)|^2 + |u'(\cdot)|^2$ is
bounded in the interval $(x_0,\infty)$ which we set out to prove.
\end{proof}
%------
Now we are ready to describe the essential spectrum of
$\opr{H}_\Lambda$.
%------
\begin{theorem}
For $\opr{H}_\Lambda$ defined by (\ref{eq:opr}) the following is
true:
%------
\begin{enumerate}
\item[(i)] For any gap $(E_{2k-1},\,E_{2k})$ in the essential
spectrum of $\opr{h}_\Lambda $,
%------
\begin{enumerate}
\item $\opr{H}_\Lambda$ has no continuous spectrum in $
(E_{2k-1},\,E_{2k}) $, and \item the point spectrum of
$\opr{H}_{\Lambda}$ is dense in $(E_{2k-1},\,E_{2k})$.
\end{enumerate}
%------
\item[(ii)] On any compact $K$ contained in the interior of a band
of $\opr{h}_\Lambda$ the spectrum of $\opr{H}_\Lambda$ is purely
absolutely continuous.
\end{enumerate}
%------
\end{theorem}
%------
\begin{proof}
(i) By (\ref{eq:EssSpectra}), none of the operators
$\opr{H}_{l},\, l=0,\,1,\,2,\dots$ has a continuous spectrum in
$(E_{2k-1},\,E_{2k})$, hence $\opr{H}_{\Lambda}$ has no continuous
spectrum in this interval either. On the other hand, the entire
interval $(E_{2k-1},\,E_{2k})$ is contained in the essential
spectrum of $\opr{H}_\Lambda$; it follows that the spectrum of
$\opr{H}_\Lambda$ in $(E_{2k-1},\,E_{2k})$ consists solely of
eigenvalues which are necessarily dense in that interval. \\
(ii) The claim follows from the previous lemma and
\cite{GilPea,Simon}. To make the article self-contained we prove
in the appendix \ref{app:A}  a weaker result which still
guarantees the absolute continuity of the spectrum in the bands in
our singular case.
\end{proof}

\setcounter{equation}{0}
%%%%%%%%%%%%%%%%%%%%%%%%%%%%%%%%%%%%%%%%%%%%%%%%%%
\section{The discrete spectrum}
\label{sec:disc_spectrum}

Recall that with the exception  of the case $\nu=2,\,l=0$ the
centrifugal term in the partial waves operators (\ref{eq:ParOpr})
is strictly positive, hence by the mini-max principle there is no
discrete spectrum below $E_0$. On the other hand, in the
two-dimensional case Brown et al. noticed that regular radially
periodic potentials give rise to bound states \cite{Welsh} which
they named in a nationalist spirit. Subsequently Schmidt
\cite{Schmidt2} proved that there are infinitely many such
eigenvalues of the operator $\opr{H}_0$ below
$\inf\sigma_\mathrm{ess} (\opr{H}_\Lambda)$. Our aim is to show
that this result persists for singular sphere interactions
considered here.

%------
\begin{theorem}
Let $\nu=2$, then except of the free case the operator
$\opr{H}_\Lambda$ has infinitely many eigenvalues in
$(-\infty,E_0)$, where $E_0:= \inf\sigma_\mathrm{ess}
(\opr{H}_\Lambda)$.
\end{theorem}
%------

%------
\begin{proof} The argument is again similar to that of the regular
case \cite{Schmidt2}, hence we limit ourselves to just sketching
it. First of all, it is clear that we have to investigate the
spectrum of $\opr{H}_{\Lambda,0}$.

Let $u,\,v$ be linearly independent real-valued solutions of the
equation $\opr{h}_\Lambda z = E_0 z$, where $u$ is (anti)periodic
--- cf.~Proposition \ref{def:Sym}. --- satisfying $W[u,\,v] = 1$.
We will search the solution of $\opr{H}_0 y \equiv -y''
-\frac{1}{4r^2} y = E_0y$, we are interested in, using a
Pr\"ufer-type Ansatz, namely
%------
\[
\left(\begin{array}{c}
    y \\ y'
\end{array}\right)
= \left(\begin{array}{cc}
 u & v  \\
 u'& v' \end{array}\right)
 a
 \left( \begin{array}{c}
   \sin \gamma \\
   -\cos \gamma
 \end{array} \right),
\]
%------
where $a$ is a positive function and $\gamma$ is chosen continuous
recalling Lemma~\ref{def:Wrons} and eq.~(\ref{eq:Wcon}). It is
demonstrated in \cite{Schmidt2} that the function $\gamma(\cdot)$
and the standard Pr\"ufer variable $\theta(\cdot)$, appearing in
%-----
\[
\left(\begin{array}{c}
    y \\ y'
\end{array}\right)
=
 \rho
 \left( \begin{array}{c}
   \cos \theta \\
   \sin \theta
 \end{array} \right),
\]
%----
are up to constant asymptotically equal to each other as $r \to
\infty$.  According to Corollary~\ref{def:osccol} there are then
infinitely many eigenvalues below $E_0$ if $\theta$, and therefore
also $\gamma$, is unbounded from below.

Now a straightforward computation yields
%----
\[
\gamma'= -\frac{1}{4r^2}(u \sin \gamma - v \cos \gamma)^2= -
\frac14\, \cos^2\gamma\, u^2 \Big(\frac{1}{r} \tan \gamma -
\frac{v}{r\,u} \Big)^2.
\]
%----
Furthermore, the Kepler transformation given by the relation $\tan
\phi = (r^{-1} \tan \gamma - r^{-1}v/u)$ satisfies $\gamma(r) =
\phi(r) + \mathcal{O}(1)$ as $r \to \infty$, and
%----
\begin{multline}
\phi' = \frac{1}{r} \Big(-\sin \phi \cos \phi -\frac14 u^2 \sin^2
\phi - \frac{1}{u^2} \cos^2 \phi \Big) \\
= -\frac{1}{2r} \bigg(\frac{1}{u^2} +\frac14 u^2 + \sin 2 \phi +
\Big(\frac{1}{u^2} -\frac14 u^2\Big) \cos 2 \phi \bigg)
\label{eq:phi}
\end{multline}
%----
holds on $\mathbb{R} \setminus \cup_n\{r_n\}$ with the
discontinuity
%----
\begin{equation}
\tan \phi(r_n+) - \tan \phi(r_n-) = -\frac{1}{r_n}
\frac{\beta}{u(r_n+) u(r_n-)}, \label{eq:discon}
\end{equation}
%----
where $\beta$ is the parameter appearing in (\ref{eq:BoundCon}). A
direct analysis of the equation (\ref{eq:phi}) shows that $\phi'
\leq 0$, and owing to (\ref{eq:discon}) and
Proposition~\ref{def:Sym} the corresponding discontinuity is
strictly negative for $\beta \neq 0$. Hence $\phi$ is decreasing
and there is a limit $L = \lim_{r \to \infty} \phi(r)$. Suppose
that $L$ is finite. Then the condition $| \int_0^\infty
\phi'(t)\,\d t| < \infty $ gives
%----------
\begin{equation}
\frac{1}{u^2(r)} + \frac{1}{4} u^2(r) + \sin2\phi(r) +
\left(\frac{1}{u^2(r)} - \frac{1}{4} u^2(r)\right)\cos2 \phi(r) \to
0 \quad \mbox{as} \quad r \to \infty \label{eq:last}
\end{equation}
%-----------
and, as $u$ is (anti)-periodic and $\phi$ tends to a constant, we
infer that $u^2$ is constant also, not only asymptotically but
everywhere. With the exception of the free case this may happen
only for pure repulsive $\delta'$ interaction, $\beta
> 0,\,\alpha = 0,\,\gamma=\delta =1$. To finish the proof we
employ eq.~(\ref{eq:last}) again and observe that $L \neq \pi/2\,
(\mathrm{mod}\, \pi)$ holds necessarily. We thus find a monotonous
sequence of points $r_{n}$ such that $\phi(r_{n}-) < \frac{\pi}{2}
\left( 1+ \left[ \frac{2L}{\pi} \right] \right)$, where $[\cdot]$
is the integer part. Since $\phi$ is monotonous we have
$\phi(r_{n}\pm) \ge L$, hence all these points belong to the same
branch of the $\tan$ function. Summing then the discontinuities
(\ref{eq:discon}) we get
%----
\begin{eqnarray*}
\tan\phi(r_N+) - \tan\phi(r_n-) &&\leq \sum\limits_{i=n}^N
\tan\phi(r_i+) - \tan\phi(r_i-) \\ && = -\sum\limits_{i=n}^N
\frac{1}{r_i} \frac{\beta}{u(r_i+)u(r_i-)}\,,
\end{eqnarray*}
%----
where the right-hand side diverges as $N \to \infty$ for any
$\beta>0$, while the left-hand side tends to a finite number
$\tan(L) - \tan \phi(r_n-)$. Hence $L$ can be finite for the free
Hamiltonian only, which was to be demonstrated.
\end{proof}

%\appendix
%\setcounter{equation}{0} \setcounter{definition}{0}
%\renewcommand{\thedefinition}{\textbf{A}-\arabic{definition}}
\setcounter{equation}{0}
%%%%%%%%%%%%%%%%%%%%%%%%%%%%%%%%%%%%%%%%%%%%%%%%%%
\appendix
\section{Continuous spectra for one dimensional \\ Schr\"{o}dinger
operators with singular interactions} \label{app:A}

In this appendix we consider Schr\"{o}dinger operators on a
halfline,
%--------
\begin{eqnarray}
 &&(\opr{H}u)(x) = -u''(x) + V(x)u(x) \\
 &&u(0)=0,\quad U(x_n+) = \Lambda U(x_n-),
\end{eqnarray}
where we suppose that the condition
%-------
\begin{equation}
 \int_{K}^\infty|u'|^2 \leq \beta \int_K^{\infty}(|\opr{H}u|^2 + |u|^2),
 \label{eq:SobCon}
\end{equation}
%-------
holds for some $\beta,\,K >0$ and every $u \in D(\opr{H})$. This is
obviously the case of operators $\opr{H}_{\lambda,\,l}$, where in
the dimension $\nu > 2$ we may put $K=0$, while for $\nu = 2$ we
have to choose $K>0$.

Given a solution $u$ of $\opr{H}u = Eu$ we define the transfer
matrix $\opr{T}(E,\,x,\,y)$ at energy $E$ by
%-----------
\begin{equation}
 \opr{T}(E,\,x,\,y) \left(\! \begin{array}{c}
                        u'(y) \\
                        u(y)
                \end{array} \! \right)
                = \left(\! \begin{array}{c}
                        u'(x) \\
                        u(x)
                \end{array} \! \right).
\label{eq:Transfer}
\end{equation}
%-----------
Our purpose is to prove the following result.

\begin{theorem} \label{def:A1}
Let $\opr{T}(E,\,x,\,y)$ be bounded on $S$. Then for every
interval $(E_1,\,E_2) \subset S $ we have $\rho_{ac}((E_1,\,E_2))
>0$ and $\rho_{sc}((E_1,\,E_2)) =0$, where $\rho$ denotes the
spectral measure associated with the operator $\opr{H}$.
\end{theorem}

\noindent Following \cite{Simon} we employ the theory of Weyl
m-functions. For $E \in C_+=\{z,\,\im z>0\}$, there is a unique
solution $u_+(x,\,E)$ of $ \opr{H} u_+(x,\,E) = E u_+(x,\,E)$ with
$u_+ \in L^2$ at infinity, which is normalized by $u_+(0,\,E)=1$.
We define the m-function by
 %---------
 $$ m_+(E) = u_+'(0,\,E)\,;$$
 %---------
the spectral measure $\rho$ is then related to it by
%-------
$$ \d \rho(E) = \frac{1}{\pi} \lim_{\varepsilon \downarrow 0}\, \im m_+(E
+ i\varepsilon)\,,$$
%-------
where the imaginary part at the right-hand side can be expressed
as
%---------
\begin{equation} \im m_+(E) = \im E \int_0^\infty |u_+(x,\,E)|^2 \d x.
\label{eq:im}
\end{equation}
%--------
It is known, see \cite{Simon} and references therein, that
%--------
$$ \mathrm{supp\,} \rho_{sc}= \Big\{E:\: \lim_{\varepsilon
\downarrow 0} \im m_+(E+i\varepsilon) = \infty \Big\}\,, $$
%--------
while $\d \rho_{ac}(E) = \frac{1}{\pi} \im m_+(E+i0)\, \d E$.
Theorem \ref{def:A1} is then an immediate consequence of the
following result.

\begin{theorem}
If $\opr{T}(E,\,x,\,y)$ be bounded as above and $E \in
(E_1,\,E_2)$, then
%--------
 $$
\lim\inf\, \im m_+(E+i0) > 0 \quad \mathrm{and} \quad \lim\sup\,
\im m_+(E+i0) <\infty\,.
 $$
%--------
\end{theorem}
%--------
\begin{proof}
For $x \neq x_n$ we have the relations
 %------
 \begin{eqnarray} && \frac{\d \opr{T}(E,\,x,\,y)}{\d x} = \left(\! \begin{array}{cc}
                                                0 & V(x) - E \\
                                                1 & 0
                                            \end{array} \! \right)
                                            \opr{T}(E,\,x,\,y),  \nonumber \\
                  && \frac{\d}{\d y}\left((\opr{T}(E_1,\,x,\,y)
                  \opr{T}(E_2,\,y,\,x)\right) =(E_1 - E_2)\opr{T}(E_1,\,x,\,y)
                  \left(\!\begin{array}{cc}
                    0 & 1 \\
                    0 & 0
                  \end{array}\!\right) \opr{T}(E_2,\,y,\,x) .\nonumber
 \end{eqnarray}
 %------
It is straightforward to verify that $\opr{T}(E_1,\,x,\,y)
\opr{T}(E_2,\,y,\,x)$ is continuous at singular points with respect
to $y$ and hence
%--------
$$1 -\opr{T}(E_1,\,x,\,0) \opr{T}(E_2,\,0,\,x) = \int_0^x (E_1 - E_2)\opr{T}(E_1,\,x,\,y)
                  \left(\!\begin{array}{cc}
                    0 & 1 \\
                    0 & 0
                  \end{array}\!\right) \opr{T}(E_2,\,y,\,x) \d y.$$
%--------
Now we put $E_1 = E,\,E_2=E+i\varepsilon$ and multiply by $\opr{T}(E
+ i\varepsilon,\,x,\,0)$ from the right to get the formula
%--------
$$ \opr{T}(E+i\varepsilon,\,x,\,0) = \opr{T}(E,\,x,\,0) -(i\varepsilon)\int_0^x \opr{T}(E,\,x,\,y)
                  \left(\!\begin{array}{cc}
                    0 & 1 \\
                    0 & 0
                  \end{array}\!\right) \opr{T}(E +
                  i\varepsilon,\,y,\,0) \d y.
                  $$
%--------
By assumption we have $\|\opr{T}(E,\,x,\,y)\| \leq C$, and
therefore
%-------
$$
\|\opr{T}(E+i\varepsilon,\,x,\,0)\| \leq C + \varepsilon\int_0^x C
\|\opr{T}(E+i\varepsilon,\,y,\,0)\| \d y\,, $$
%-------
so by iteration we get
%-------
$$
\|\opr{T}(E+i\varepsilon,\,x,\,0)\| \leq C e^{\varepsilon C x}.
$$
%-------
Note that $\det \opr{T} =1$ so $\|\opr{T}\| = \|\opr{T}^{-1}\|$.
Putting now $\gamma = ((E+1)^2 \beta^2 + 1)^{-1}$ and using the
condition (\ref{eq:SobCon}) we get
%--------
\begin{align}
\int_0^\infty|u(x)|^2 \d x &\geq \gamma \int_K^\infty(|u(x)|^2 +
|u'(x)|^2) \d x  \nonumber \\
 &\geq C^{-2} \gamma (1 + |m_+|^2)
\int_K^\infty e^{-2 \varepsilon C x} \d x\,, \nonumber
\end{align}
%--------
hence by (\ref{eq:im}) we infer that
%--------
$$\im m_+ \geq \frac{1}{2} C^{-3} \gamma (1 + |m_+|^2) \,.$$
%-------
From here the first claim follows immediately, and since
%----
$$ 2 C^3 \gamma^{-1} \geq \frac{1 + |m_+|^2}{\im m_+} \geq
|m_+|\,,
$$
%----
we get also the remaining part.

\end{proof}

\section{Oscillation theory for singular potentials}
\label{app:osc}

In the case of point interactions the classical oscillation theory
fails due to discontinuity of the wave functions. Nevertheless, we
can employ the continuity of the Wronskian and formulate the
oscillation theory using the approach of relative oscillations
\cite{Simon2}. The aim of this appendix is to present briefly the
basic theorems; since the claims are the same as in the regular
case we follow closely the above mentioned article.

We consider Schr\"odinger-type operators on $L^2(l_-,\,l_+)$ with
the singular interactions at the points $x_n \in (l_-,\,l_+),\,n
\in M \subset \mathbb{N}$ which act as
%-------
\[
\opr{T} u(x) = -u''(x) + q(x)u(x),
\]
%------
with a real-valued potential $q \in L^1_{\mathrm{loc}}(l_-,\,l_+)$
and the domain
%------
 \begin{multline}
 D(\opr{T}) = \bigg\{u,\,u' \in
AC_{\mathrm{loc}}(l_-,\,l_+) \setminus \bigcup_{n \in M}\{x_n\}\::
\\ \,Tu \in L^2_{\mathrm{loc}}(l_-,\,l_+) \quad \mbox{and} \quad
U(x_n+) = \Lambda_n U(x_n-) \bigg\}\,. \nonumber
\end{multline}
%------
Such an operator is obviously symmetric. Denote by $\opr{H}$ an
arbitrary self-adjoint extension of it satisfying either
%------
\begin{enumerate}
  \item[(a)] $\opr{T}$ is limit point in at least one endpoint, or
  \item[(b)] $\opr{H}$ is defined by separated boundary conditions.
\end{enumerate}
%------
By $\psi_{\pm}(E,\,x)$ we denote real-valued solutions of the
equation $\opr{T} \psi_{\pm}(E,\,x) = E\psi_{\pm}(E,\,x)$, which
satisfy the boundary conditions defining $\opr{H}$ at the points
$l_{\pm}$, respectively. Note that such solutions may not exist, the
theorems given below implicitly assume their existence. In
particulary, their existence is assure for energies $E$ outside the
essential spectrum. And with respect to analyticity in spectral
parameter we may use the oscillation theory also at the edge of the
essential spectrum.

The first theorem to follow provides the basic oscillation result,
while the corollary of the second one is the result used in
Section~\ref{sec:disc_spectrum}. By $W_0(u_1,\,u_2)$ we denote the
number of zeros of the Wronskian $W[u_1,\,u_2](x)$ in the open
interval $(l_-,\,l_+)$, and given $E_1 < E_2$, we put
$N_0(E_1,\,E_2) = \mathrm{dim}\,\mathrm{Ran}
\opr{P}_{(E_1,\,E_2)}$, where $\opr{P}$ is a spectral measure of
the self-adjoint operator $\opr{H}$. In particular, in case of the
pure point spectrum $N_0(E_1,\,E_2)$ simply denotes the number of
eigenvalues in the interval $(E_1,\,E_2)$.

\begin{theorem}
  Suppose that $E_1 < E_2$ and put $u_1 = \psi_-(E_1),\; u_2=\psi_+(E_2)$.
  Then $W_0(u_1,\,u_2) = N_0(E_1,\,E_2)$.
\end{theorem}

\begin{theorem}
  Let $E_1 < E_2$. Assume that either $u_1 = \psi_+(E_1)$ or $u_1 =
  \psi_-(E_1)$ holds, and similarly either $u_2 = \psi_+(E_2)$ or $u_2 =
  \psi_-(E_2)$. Then $W_0(u_1,\,u_2) \leq N_0(E_1,\,E_2)$.
  \label{def:t2osc}
\end{theorem}
Next we introduce Pr\"{u}fer variables $\rho_i,\,\theta_i$ defined
by
%----
\[
\left( \begin{array}{c}
            u_i(x) \\
            u_i'(x)
       \end{array} \right) = \rho_i(x)
       \left(\begin{array}{c}
            \cos \theta_i(x) \\
            \sin \theta_i(x)
       \end{array}\right),
\]
%----
where $\rho_i$ is chosen positive and $\theta_i$ is uniquely
determined by its boundary value and the requirement that
$\theta_i$ is continuous on $(l_-,\,l_+) \setminus \bigcup_{n \in
M}  \{x_n\}$ while its discontinuity at the sites $x_n$ of the
point interactions satisfies $|\theta_i(x_n+) - \theta_i(x_n-)|
=0\: (\mathrm{mod}\, \pi)$.
%-----
\begin{corollary}
  \label{def:osccol}
  Suppose that $E_1$ is the edge of the essential spectrum, and $u_1 =
  \psi_-(E_1)$ or $u_1 = \psi_+(E_1)$. Then $\opr{H}$ has infinitely many
  eigenvalues below $E_1$ if $\theta_1(\cdot)$ is unbounded.
\end{corollary}
%-----
\begin{proof}
In analogy with the regular case the function
$\theta_2$corresponding to $u_2=\psi_{\pm}(E)$ is bounded for
negative $E$ large enough. This implies that $|\theta_2 -
\theta_1| \to \infty$ and since $W[u_1,\,u_2](x) = \rho_1(x)
\rho_2(x) \sin (\theta_2(x) - \theta_1(x)) $ we get
$W_0(u_1,\,u_2) = \infty$. Hence Theorem \ref{def:t2osc}.
completes the proof.
\end{proof}

\subsection*{Acknowledgment}

The research was supported by the Czech Ministry of Education,
Youth and Sports within the project LC06002.

\end{document}